\documentclass[journal,twoside]{IEEEtran}
\usepackage{enumerate}
\usepackage{graphicx}
\usepackage{mathrsfs}
\usepackage{amssymb,amsfonts,amsmath,bm}
\usepackage{amsthm}
\usepackage{tabu}

\newtheorem{definition}{Definition}[section]
\newtheorem{theorem}[definition]{Theorem}

\newtheorem{lemma}[definition]{Lemma}

\begin{document}
\title{A Novel Construction of Low-Complexity MDS Codes with Optimal Repair Capability for Distributed Storage Systems}
\author{Sheng~Guan,
        Haibin~Kan,~\IEEEmembership{Member,~IEEE,}
        and Xin~Wang,~\IEEEmembership{Senior Member,~IEEE}
\thanks{S. Guan and H. Kan are with the Shanghai Key Laboratory of Intelligent Information Processing, Fudan-Zhongan Joint Lab of Blockchain and Information Security, School of Computer Science, Fudan University, Shanghai 200433, P. R. China (e-mail:hbkan@fudan.edu.cn). }
\thanks{X. Wang is with the School of Information Science and Engineering, Fudan University, Shanghai 200433, P. R. China. }
}
\maketitle
\begin{abstract}
Maximum-distance-separable (MDS) codes are a class of erasure codes that are widely adopted to enhance the reliability of distributed storage systems (DSS). In $(n,k)$ MDS coded DSS, the original data are stored into $n$ distributed nodes in an efficient manner such that each storage node only contains a small amount (i.e., $1/k$) of the data and a data collector connected to any $k$ nodes can retrieve the entire data. On the other hand, a node failure can be repaired (i.e., stored data at the failed node can be successfully recovered) by downloading data segments from other surviving nodes. 
In this paper, we develop a new approach to construction of simple $(5,3)$ MDS codes. With judiciously block-designed generator matrices, we show that the proposed MDS codes have a minimum stripe size $\alpha =2$ and can be constructed over a small (Galois) finite field $\mathbb{F}_4$ of only four elements, both facilitating low-complexity computations and implementations for data storage, retrieval and repair. In addition, with the proposed MDS codes, any single node failure can be repaired through interference alignment technique with a minimum data amount downloaded from the surviving nodes; i.e., the proposed codes ensure optimal exact-repair of any single node failure using the minimum bandwidth. The low-complexity and all-node-optimal-repair properties of the proposed MDS codes make them readily deployed for practical DSS.
\end{abstract}


\IEEEpeerreviewmaketitle

\section{Introduction}

\begin{table*}
\centering
\caption{Comparison of different codes in terms of stripe size, field size, and node-repair capability.}
\begin{tabu} to \hsize {|X|X|X|X|X|X|X|}
\hline
Codes & $(5,3)$ Hadamard MDS codes \cite{Had} & $(5,3)$ Zigzag MDS codes \cite{Zig} & $(5,3)$ MSR codes \cite{Product} & $(6,3)$ MDS codes \cite{Interference} & $(6,4)$ MDS codes \cite{LMDS} & Proposed $(5,3)$ MDS codes \\ \hline
Stripe size & 16 & 4 & 2 & 3 & 2 & 2 \\ \hline
Field size & $\geq 9$ & $\geq 3$ & $\geq 10$ & $\geq 6$ & $\geq 4$ & $\geq 4$ \\ \hline
Optimal repair & Yes & No & Yes & No & No & Yes\\
\hline
\end{tabu}
\end{table*}

To ensure the reliability of distributed storage systems (DSS), erasure codes can be employed to encode the data segments across the storage nodes. Relying on erasure codes, a data collector can reconstruct the entire data in the presence of (multiple) node failures in DSS.
Maximum-distance-separable (MDS) codes are a class of erasure codes. To facilitate low-complexity implementation, systematic MDS codes are widely adopted for DSS. Specifically, if $(n,k)$ systematic MDS codes are used, then the entire raw data of $M$ symbols are divided into $k$ segments of $\alpha=M/k$ symbols. The $k$ (raw) data segments are directly stored in $k$ ``systematic'' nodes, and linear combinations of these $k$ segments are stored in $n-k$ ``parity'' nodes, such that a data collector connected to any $k$ nodes can reconstruct the entire data. In other words, data retrieval for the MDS coded DSS can tolerate no more than $n-k$ node failures. Fig.1 presents an example of $(4,2)$ MDS coded DSS, where it can be readily verified that the entire data can be reconstructed by a data collector connected to every two nodes, and every single node failure can be repaired by a newcomer (i.e., a replacement) downloading data from other three nodes.

\begin{figure}
\centering
\includegraphics[width=2in]{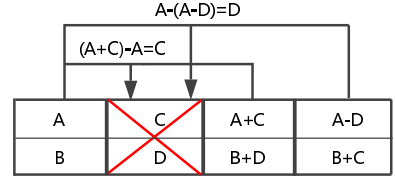}
\caption{An example of $(4,2)$ MDS codes: When node 2 fails, its data can be recovered by downloading data from other three nodes. Specifically, symbol C can be recovered by subtracting symbol A at node 1 from symbol A+C at node 3; and symbol D can be recovered by subtracting symbol A-D at node 4 from symbol A at node 1.}
\end{figure}

Building on MDS codes, Dimakis et al. \cite{Dss} proposed a new class of regenerating codes to achieve the efficient repair of node failures for DSS. Two special cases of regenerating codes, minimum storage regenerating (MSR) codes and minimum bandwidth regenerating (MBR) codes, were developed. It was also shown that the theoretic lower bound of the repair bandwidth (i.e., the total data amount that is required to be downloaded) with MSR codes for repairing single node failures is given by \cite{Dss}
\begin{equation} \label{lower}
    \gamma_{min} = \frac{M}{k} \cdot \frac{n - 1}{n - k}.
\end{equation}
Rashimi et al. \cite{Product} developed a product-matrix framework to construct MSR codes that allow exact-repair of every single node failure using such a minimum bandwidth. The MSR codes in \cite{Product} have a code rate of $k/n \leq k/(2k-1)$ (whose upper bound is close to $1/2$ as $k$ grows large) and need to be constructed over a finite field of size $q \geq n (n-k)$. Y. Wu \cite{Wu} further introduced a construction of functional-repair MDS codes achieving the optimal repair bandwidth, where a part of data is allowed to be changed after repair as long as the ability to retrieve the original data remains. Functional-repair in DSS is usually not preferred in practice since change of data can lead to an increasing data access time. Shah et al. \cite{Interference} relied on interference alignment techniques to construct $(2k,k)$ exact-repair MDS codes. A disadvantage of their codes is that only the systematic nodes can be optimally repaired whereas repair of a parity node requires downloading data of the original size (i.e., $M$ symbols). C. Suh and K. Ramchandran \cite{Exact-IA} improved the construction in \cite{Interference} by using the technique of common eigenvectors, so that their codes possess all-node-repair capability using minimum repair bandwidth. S. Liang et al. \cite{a} also provided a construction of all-node-optimal exact-repair $(6,3)$ MSR codes over the finite field $\mathbb{F}_2$. The code rate of the MSR codes in \cite{Interference, Exact-IA, a} is 1/2, implying that the same amount ($M$ symbols) of ``redundant'' data need to be stored to ensure the reliability of DSS.

To build more efficient DSS, \cite{Had} and \cite{Zig} proposed constructions of $(k+2,k)$ MDS codes with optimal repair bandwidth. Papailiopoulos et al. \cite{Had} used Hadamard matrices to enable perfect interference alignment in repair of node failures. Their codes can be constructed over a finite field of size $q \geq 2k+3$, and ensure that all single node failures can be optimally repaired using minimum bandwidth. On the other hand, Tamo et al. \cite{Zig} employed zigzag codes to ensure exact-repair of systematic nodes optimally, and their zigzag codes can be constructed over a small finite field $\mathbb{F}_3$. Codes in \cite{Had} and \cite{Zig} require the stripe size (i.e., the number of stored data symbols per node) $\alpha = 2^{k+1}$ and $\alpha = 2^{k-1}$, respectively. It was shown in \cite{Stripe_size} that a small stripe size facilitates flexible and low-cost operations of DSS; hence, it is preferred in practice. To this end, \cite{LMDS} developed a construction of $(6,4)$ MDS codes with the minimum stripe size $\alpha = 2$, whereas \cite{Interference} also derived a construction of $(6,3)$ MDS codes with $\alpha=3$. As with the codes in \cite{Zig}, these two constructions can only ensure optimal repair of systematic nodes.

In this paper, we develop a novel construction of $(5,3)$ MDS codes with the minimum stripe size $\alpha = 2$. With judiciously block-designed generator matrices, we show that the proposed MDS codes can be constructed over a small (Galois) finite field $\mathbb{F}_4$ of only four elements. In addition, we rigorously establish that the proposed MDS codes can ensure the repair of any single node failure through interference alignment technique using the minimum bandwidth; i.e., the proposed codes are indeed MSR codes. Tab. I compares the proposed MDS codes with the existing alternatives in terms of stripe size, (construction) field size, and node-repair capability. Note that the $(5,3)$ Hadamard based MDS codes \cite{Had} require the stripe size $\alpha = 16$ and need to be constructed over a finite field of (at least) size 9, which are much larger than the stripe size ($\alpha = 2$) and field size ($q=4$) with the proposed MDS codes. The $(5,3)$ zigzag codes \cite{Zig}, $(6,3)$ MDS codes \cite{Interference}, and $(6,4)$ MDS codes \cite{LMDS} have stripe sizes $\alpha = 4,3,2$, respectively; yet, they can only ensure optimal repair of systematic nodes. Compared to the existing codes, the proposed MDS codes have optimal repair capability for both systematic nodes and parity nodes, and facilitate low-complexity computations and implementations for data storage, retrieval and repair. They can be attractive MDS code candidates for practical DSS.

The rest of the paper is organized as follows. In Sec . II, we describe the construction of the proposed MDS codes. Sec. III proves that the proposed MDS codes ensure the optimal exact repairs of all single node failures. Sec. IV provides examples of the proposed MDS codes. Sec. V concludes the paper.

\textit{Notations:} The following notations are used throughout the paper. Boldface letters denote vectors or matrices; $^T$ denotes transpose; 
$\text{rank}(\bm{A})$ denotes the rank of a matrix $\bm{A}$; $\bm{I}$ denotes an identity matrix with a proper size; $\bm{0}$ denotes an all-zero vector (or matrix) with a proper size; $\mathbb{F}_q$ denotes a Galois finite field of $q$ elements, and when it is clear from the context, addition, subtraction, multiplication and division are defined over Galois finite field.


\section{Code Construction}

In this section, we delineate the explicit construction of our $(5,3)$ MDS codes with a stripe size of 2; i.e., here we have $n=5$, $k=3$, and $\alpha=2$.
The proposed codes are used for DSS with $n=5$ nodes to store a data file containing $M=k\alpha=6$ symbols over a certain finite field $\mathbb{F}_q$. Each node needs to store a (coded) data segment of $\alpha =2$ symbols. Note that one symbol can be a data block of many (e.g., 4000) bits.

Denote the entire data (file) by a $6\times 1$ vector $\bm{m} := [\bm{m}_1^{T},\; \bm{m}_2^{T}, \;\bm{m}_3^{T}]^{T}$, where $\bm{m}_i$, $i=1,2,3$, are $2\times 1$ vectors. Similar to \cite{Had, LMDS}, we adopt a $n\alpha \times k\alpha$ (i.e., $10 \times 8$) generator matrix in the form:
\begin{equation} \label{P}
\bm{P} = \left[
        \begin{array}{ccc}
        \bm{I} & & \\
         & \bm{I} &\\
         & & \bm{I}\\
        \bm{I} & \bm{I} & \bm{I}\\
        \bm{A}_1^T & \bm{A}_2^T & \bm{A}_3^T\\
        \end{array}
\right]
\end{equation}
where $\bm{I}$ denotes a $2 \times 2$ identity matrix, and $\bm{A}_i$, $i=1,2,3$, are $2 \times 2$ matrices to be designed. With the generator matrix $\bm{P}$, the encoded data can be collected into a $10 \times 1$ vector $\bm{u} := [\bm{u}_1^{T},\; \bm{u}_2^{T}, \;\bm{u}_3^{T}, \; \bm{u}_4^{T}, \; \bm{u}_5^{T}]^{T}$, where $\bm{u}_i$, $i=1,\ldots,5$, are $2\times 1$ vectors. Then we have
\begin{equation} \label{U}
\begin{array}{l}
\bm{u} = \bm{P} \cdot \bm{m} \\ = \left[
        \begin{matrix}
        \bm{m}_1\\
        \bm{m}_2\\
        \bm{m}_3\\
        \bm{m}_1 + \bm{m}_2 + \bm{m}_3\\
        \bm{A}_1^T \bm{m}_1 + \bm{A}_2^T \bm{m}_2 + \bm{A}_3^T \bm{m}_3\\
        \end{matrix}
\right] = \left[
\begin{matrix}
        \bm{u}_1\\
        \bm{u}_2\\
        \bm{u}_3\\
        \bm{u}_4\\
        \bm{u}_5\\
        \end{matrix}
\right].
\end{array}
\end{equation}
The symbols in $\bm{u}_i$ are exactly the data stored in node $i=1,\ldots,5$. Clearly, data stored in the first $k=3$ nodes are in raw form; hence, these nodes are called systematic nodes. On the other hand, the rest $2$ nodes are called parity nodes as their data are linear combinations of the raw data segments.

Construction of the proposed MDS codes boils down to designing proper matrices $\bm{A}_1$, $\bm{A}_2$, and $\bm{A}_3$ such that a data collector connected to any $k=3$ nodes can retrieve the entire original data file, and any single node failure can be repaired using minimum bandwidth.

\subsection{Construction of $\bm{A}_1$, $\bm{A}_2$, and $\bm{A}_3$ over $\mathbb{F}_4$}

Denote the 2-dimensional standard basis vectors by $\bm{e}_1: =[1 \;\; 0]^T$ and $\bm{e}_2=[0 \;\; 1]^T$. Select four parameters $\lambda$, $\mu$, $\theta$ and $\eta$ over a certain finite field $\mathbb{F}_q$, and let
\begin{equation}\label{A}
\begin{array}{l}
\bm{A}_1 = \left[
    \begin{array}{c}
    \theta \bm{e}_1^T\\
    \eta \bm{e}_1^T + \bm{e}_2^T \\
    \end{array}
    \right] = \left[
    \begin{array}{cc}
    \theta & 0 \\
    \eta & 1 \\
    \end{array}
    \right], \\
\bm{A}_2 = \frac{1}{\lambda} \left[
    \begin{array}{c}
    (\theta-1) \bm{e}_1^T\\
    \eta \bm{e}_1^T + \bm{e}_2^T \\
    \end{array}
    \right] = \frac{1}{\lambda} \left[
    \begin{array}{cc}
    \theta -1 & 0 \\
    \eta & 1 \\
    \end{array}
    \right], \\
\bm{A}_3 = \frac{1}{\mu} \left[
    \begin{array}{c}
    \theta \bm{e}_1^T - \bm{e}_2^T\\
    \eta \bm{e}_1^T + \bm{e}_2^T \\
    \end{array}
    \right] = \frac{1}{\mu} \left[
    \begin{array}{cc}
    \theta & - 1 \\
    \eta & 1 \\
    \end{array}
    \right].
    \end{array}
\end{equation}
Note that there exists a relationship among $\bm{A}_1^T$, $\bm{A}_2^T$ and $\bm{A}_3^T$, which is
\begin{equation} \label{A123}
\bm{A}_1^T = \lambda \bm{A}_2^T + [\bm{e}_1 \;\; \bm{0}] = \mu \bm{A}_3^T + [\bm{e}_2 \;\; \bm{0}],
\end{equation}
where $\bm{0}$ denotes a zero vector of length $2$. Such a construction can facilitate the implementation of interference alignment to perform optimal node repair using minimum bandwidth, as will be shown in the sequel.

In order to obtain the desired MDS codes, we require that the parameters $\{\lambda, \mu, \theta, \eta\}$ satisfy the following conditions:
\begin{subequations}\label{limi}
\begin{align}
& \lambda \neq 0, 1 ; \label{limi-a}\\
& \mu \neq \lambda,0,1; \label{limi-b}\\
& \theta \neq 0, 1 ; \label{limi-c}\\
& \eta \neq 0, -1 ; \label{limi-d}\\
& \theta + \eta \neq 0 ; \label{limi-e}\\
& \theta (1- \lambda) \neq 1 ; \label{limi-f}\\
& \theta (\mu -1) \neq \eta ; \label{limi-g}\\
& \theta (\mu - \lambda) \neq \eta \lambda + \mu ; \label{limi-h}\\
& (\mu - 1) \neq \eta (1- \lambda) ; \label{limi-i}\\
& \eta +1 \neq \eta \frac{(\theta -1)(\mu - \lambda)}{\theta \lambda(\mu - 1)}. \label{limi-j}
\end{align}
\end{subequations}

\begin{table}
\centering
\caption{The product table over $\bm{F}_4$ generated by primitive polynomial $w^2 + w + 1$}
\begin{tabu} to \hsize {|X|X|X|X|X|}
\hline
 & $0$ & $1$ & $w$ & $w+1$ \\ \hline
$0$ & $0$ & $0$ & $0$ & $0$ \\ \hline
$1$ & $0$ & $1$ & $w$ & $w+1$ \\ \hline
$w$ & $0$ & $w$ & $w+1$ & $1$ \\ \hline
$w+1$ & $0$ & $w+1$ & $1$ & $w$ \\ \hline
\end{tabu}
\end{table}

Interestingly, we can show that the matrices $\bm{A}_1$, $\bm{A}_2$, and $\bm{A}_3$ in (\ref{A}) under conditions (\ref{limi-a})--(\ref{limi-j}) can be actually constructed over a finite field as small as $\mathbb{F}_4$. To see it, we use $\mathbb{F}_4$ generated by primitive polynomial $w^2 + w + 1$ for illustration. We denote elements in this $\mathbb{F}_4$ by $\{0, 1, w, w+1\}$. Note that the opposite of each element corresponds to itself; i.e., $-1$ in $\mathbb{F}_4$ equals to 1. Tab. II provides the product table for these four elements. We can establish that

\begin{theorem} \label{F4}
The proposed codes can be constructed over finite field $\mathbb{F}_4$.
\end{theorem}
\begin{proof}
In $\mathbb{F}_4$ generated by primitive polynomial $w^2 + w + 1$, let $\eta = \lambda  = w$ and $\theta = \mu = w + 1$. Then the conditions (\ref{limi-a})--(\ref{limi-d}) clearly hold. Furthermore, we have
\[
\begin{array}{l}
\theta + \eta = 1 \neq 0 , \\
\theta (1- \lambda) = (w + 1)(1 + w) = w \neq 1 , \\
\theta (\mu -1) = (w + 1) w = 1 \neq \eta , \\
\theta (\mu - \lambda) = (w + 1) \neq \eta \lambda + \mu = w \cdot w + w + 1 = 0 , \\
(\mu - 1) = w \neq \eta (1- \lambda) = w (1 + w) = 1 , \\
\eta +1 = w + 1 \neq \eta \displaystyle \frac{(\theta -1)(\mu - \lambda)}{\theta \lambda(\mu - 1)} = w . \\
\end{array}
\]
In other word, conditions (\ref{limi-e})--(\ref{limi-j}) are satisfied. It is also easy to see that with $\eta = \lambda  = w$ and $\theta = \mu = w + 1$, all entries of $\bm{A}_1$, $\bm{A}_2$, and $\bm{A}_3$ in (\ref{A}) are either 0, 1, $w$, or $w+1$.
\end{proof}

\subsection{MDS Property}

We next show that the proposed codes are in fact MDS codes. To this end, we first establish the following two lemmas for the proposed matrices $\bm{A}_1$, $\bm{A}_2$, and $\bm{A}_3$.
\begin{lemma}\label{invert}
$\bm{A}_1$, $\bm{A}_2$, and $\bm{A}_3$ have full rank.
\end{lemma}

\begin{proof}
By conditions (\ref{limi-c})--(\ref{limi-e}) and $\lambda , \mu \neq 0$ in (\ref{A}), we can readily derive
$\text{rank}(\bm{A}_1) = \text{rank}(\bm{A}_2) = \text{rank}(\bm{A}_3) = 2$.
\end{proof}

\begin{lemma}\label{MDS}
$\bm{A}_1 - \bm{A}_2$, $\bm{A}_1 - \bm{A}_3$ and $\bm{A}_2 - \bm{A}_3$ have full rank.
\end{lemma}

\begin{proof}
Since $\lambda , \mu \neq 0$, we have
\begin{equation}\nonumber
\begin{array}{l}
     \text{rank} ( \bm{A}_1 - \bm{A}_2 ) = \text{rank}
    (
    \frac{1}{\lambda}
    \left[
    \begin{array}{c}
    (\theta \lambda - \theta + 1) \bm{e}_1^T \\
    (\lambda - 1) (\eta \bm{e}_1^T + \bm{e}_2^T) \\
    \end{array}
    \right]
    ) \\
    ~~~~ = \text{rank}
    (
    \left[
    \begin{array}{cc}
    \theta \lambda - \theta + 1 & 0 \\
    \eta(\lambda - 1) & \lambda - 1 \\
    \end{array}
    \right]
    )=2
\end{array}
\end{equation}
due to $\lambda \neq 1 $ and $1 - \theta (1- \lambda) = \theta \lambda - \theta + 1 \neq 0$ by (\ref{limi-f});
\begin{equation}\nonumber
\begin{array}{l}
    \text{rank} ( \bm{A}_1 - \bm{A}_3 )  = \text{rank}
    (
    \frac{1}{\mu}
    \left[
    \begin{array}{c}
    \theta (\mu - 1) \bm{e}_1^T + \bm{e}_2^T \\
    (\mu - 1) (\eta \bm{e}_1^T + \bm{e}_2^T) \\
    \end{array}
    \right]
    ) \\
    ~~~~ =  \text{rank}
    (
    \left[
    \begin{array}{cc}
    \theta (\mu - 1) & 1 \\
    \eta (\mu - 1) & \mu - 1 \\
    \end{array}
    \right] ) \\
     ~~~~ =  \text{rank}
    (
    \left[
    \begin{array}{cc}
    \theta (\mu - 1) - \eta & 1 \\
    0 & \mu - 1 \\
    \end{array}
    \right]
    ) = 2
    \end{array}
\end{equation}
due to $ \mu \neq 1$ and $\theta (\mu - 1) - \eta \neq 0$ by (\ref{limi-g}); and
\begin{equation}
\begin{array}{l}
    \text{rank} ( \bm{A}_2 - \bm{A}_3 ) = \text{rank}
    (
    \left[
    \begin{array}{c}
    (\frac{\theta - 1}{\lambda} - \frac{\theta}{\mu}) \bm{e}_1^T + \frac{1}{\mu} \bm{e}_2^T \\
    (\frac{1}{\lambda} - \frac{1}{\mu}) (\eta \bm{e}_1^T + \bm{e}_2^T) \\
    \end{array}
    \right]
    ) \\
    ~~~~ = \text{rank}
    (
    \left[
    \begin{array}{cc}
    \frac{\theta - 1}{\lambda} - \frac{\theta}{\mu} & \frac{1}{\mu}  \\
    \eta (\frac{1}{\lambda} - \frac{1}{\mu}) & \frac{1}{\lambda} - \frac{1}{\mu} \\
    \end{array}
    \right]
    ) \\
    ~~~~ = \text{rank}
    (
    \left[
    \begin{array}{cc}
    \frac{\theta - 1}{\lambda} - \frac{\theta + \eta}{\mu} & \frac{1}{\mu}  \\
    0 & \frac{1}{\lambda} - \frac{1}{\mu} \\
    \end{array}
    \right]
    ) = 2
    \end{array}
\end{equation}
due to $\lambda \neq \mu$ and $\theta (\mu - \lambda) - \eta \lambda - \mu = \lambda \mu (\frac{\theta - 1}{\lambda} - \frac{\theta + \eta}{\mu}) \neq 0$ by (\ref{limi-h}).
\end{proof}

Based on Lemmas II.2 and II.3, we can then show
\begin{theorem}
The proposed codes have MDS property.
\end{theorem}
\begin{proof}
We prove the lemma by checking whether the proposed codes can ensure that the entire data can be retrieved by using data segments at arbitrary three nodes. With the proposed codes, if a data collector connects to the three systematic nodes or any two systematic nodes and the first parity node, then clearly the entire data can be readily retrieved. Suppose that the data collector connects to the first and the second systematic nodes, as well as the second parity node. Then it can obtain
\begin{equation}\label{eq.Am}
\left[
\begin{matrix}
        \bm{u}_1\\
        \bm{u}_2\\
        \bm{u}_5\\
        \end{matrix}
\right] = \left[
        \begin{matrix}
        \bm{I} & & \\
         & \bm{I} &\\
        \bm{A}_1^T & \bm{A}_2^T & \bm{A}_3^T\\
        \end{matrix}
\right] \bm{m}.
\end{equation}
It follows from Lemmas II.2 that $\bm{A}_3$ has full rank; hence, the big matrix before $\bm{m}$ in (\ref{eq.Am}) is invertible, and the original data $\bm{m}$ can be retrieved. Similarly, as $\bm{A}_2$ (or $\bm{A}_1$) has full rank by Lemmas II.2, the original data $\bm{m}$ can be retrieved when the data collector connects to the first (or second) and the third systematic nodes, as well as the second parity node.
Furthermore, as $\bm{A}_1 - \bm{A}_2$, $\bm{A}_1 - \bm{A}_3$ and $\bm{A}_2 - \bm{A}_3$ have full rank by Lemmas II.3, we can show that the original data $\bm{m}$ can be retrieved when the data collector connects to one of the systematic nodes and two parity nodes. As the entire data can be retrieved by the data collector connected to arbitrary three nodes, the proposed codes are MDS codes.
\end{proof}

Theorems II.1 and II.4 establish that the proposed codes are MDS codes that can be constructed over a small finite field $\mathbb{F}_4$. With such a construction, the computational complexity of retrieving the original data file is low. It is clear that data retrieval with the three systematic nodes consumes no computation. For retrieval with two systematic nodes and one parity node, take (\ref{eq.Am}) as an example. In this case, we have
\begin{align}
\bm{m} & = \left[
        \begin{matrix}
        \bm{I} & & \\
         & \bm{I} &\\
        \bm{A}_1^T & \bm{A}_2^T & \bm{A}_3^T\\
        \end{matrix}
\right]^{-1} \left[
\begin{matrix}
        \bm{u}_1\\
        \bm{u}_2\\
        \bm{u}_5\\
        \end{matrix}
\right] \nonumber \\
& =  \left[
        \begin{matrix}
        \bm{I} & & \\
         & \bm{I} &\\
        -\bm{A}_3^{-T}\bm{A}_1^T & -\bm{A}_3^{-T}\bm{A}_2^T & \bm{A}_3^{-T}\\
        \end{matrix}
\right] \left[
\begin{matrix}
        \bm{u}_1\\
        \bm{u}_2\\
        \bm{u}_5\\
        \end{matrix}
\right] \nonumber \\
& = \left[
\begin{matrix}
        \bm{u}_1\\
        \bm{u}_2\\
        -\bm{A}_3^{-T}\bm{A}_1^T\bm{u}_1 - \bm{A}_3^{-T}\bm{A}_2^T \bm{u}_2 + \bm{A}_3^{-T}\bm{u}_5\\
        \end{matrix}
\right].\nonumber
\end{align}
Recovering $\bm{m}$ only requires calculating $-\bm{A}_3^{-T}\bm{A}_1^T\bm{u}_1 - \bm{A}_3^{-T}\bm{A}_2^T \bm{u}_2 + \bm{A}_3^{-T}\bm{u}_5$. Note that the coefficient matrices $-\bm{A}_3^{-T}\bm{A}_1^T$, $\bm{A}_3^{-T}\bm{A}_2^T$, and $\bm{A}_3^{-T}$ can be pre-computed offline.\footnote{Even if they need to be computed online, the computational complexity is low as $\bm{A}_1$, $\bm{A}_2$, and $\bm{A}_3$ are all $2 \times 2$ matrices.} Hence, the computation for online retrieval consists of $k\alpha^2$ multiplications and $(2k-1)\alpha$ additions. With $M=k\alpha$, the overall complexity is ${\cal O}(M \alpha)$. For the worst case where two parity nodes are involved to reconstruct the original file, it was shown that the data retrieval can be still performed with a computational complexity of ${\cal O}(M \alpha)$. With $\alpha=2$ in the proposed codes, such a complexity is clearly low.
In addition, the proposed codes are constructed over a small finite field $\mathbb{F}_4$; hence, the computational complexity is indeed very low.

\section{Node Repair through Interference Alignment}

\begin{figure}
\centering
\includegraphics[width=2.5in]{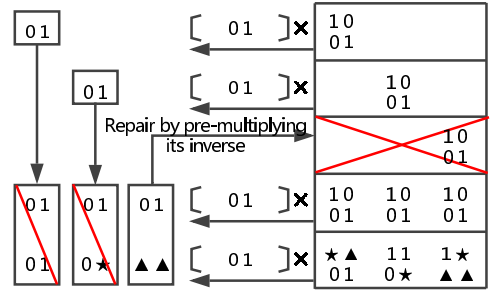}
\caption{Repair of the third systematic node}
\end{figure}

In this paper, we mainly concentrate on the repair of a single node failure since such a case occurs at most of time in practice.
In this section, we rigorously establish that the proposed MDS codes have optimal repair capability; i.e., they are MSR codes that can ensure exact repair of every single node failure using minimum bandwidth. From  (\ref{lower}), we can derive the minimum repair bandwidth in our scenario as
\begin{equation} \label{lower1}
\gamma_{min} = \frac{M}{k} \cdot \frac{n - 1}{n - k} = \alpha \cdot \frac{k + 1}{2}= 4~(\text{symbols}).
\end{equation}
To achieve this optimal repair bandwidth, a key technique that we employ for node repair is interference alignment.

\subsection{Repair of Systematic Nodes}

The interference alignment technique was first proposed to approach the capacity of interference channels in the context of wireless communications \cite{Int_1, Int_2}, and it was later introduced into MDS code design for DSS by Y. Wu et al. \cite{IA}. We first show that optimal repair of systematic nodes can be achieved with the proposed MDS codes through interference alignment.

Assume that the $i$th ($i=1,2,3$) systematic node fails. By (\ref{lower1}), repair of a node failure is allowed to download only one symbol per node from the four surviving nodes. Let $\bm{\varphi}_{i1}^T$ and $\bm{\varphi}_{i2}^T$ denote two $1 \times 2$ vectors. From the first and second parity nodes, the newcomer can download one symbol $d_{i1}$ and one symbol $d_{i2}$, respectively, as follows.
\begin{align}
\left[
\begin{array}{c}
d_{i1} \\
d_{i2} \\
\end{array}
\right]&  =
\left[
\begin{array}{c}
\bm{\varphi}_{i1}^T \bm{u_{k+1}} \\
\bm{\varphi}_{i2}^T \bm{u_{k+2}} \\
\end{array}
\right] \nonumber \\
& =
\left[
\begin{array}{c}
\bm{\varphi}_{i1}^T \\
\bm{\varphi}_{i2}^T \bm{A}_i^T \\
\end{array}
\right] \bm{m}_i + \sum_{j \neq i}
\left[
\begin{array}{c}
\bm{\varphi}_{i1}^T \\
\bm{\varphi}_{i2}^T \bm{A}_j^T \\
\end{array}
\right] \bm{m}_j. \label{Sysrep}
\end{align}
Note that symbols $d_{i1}$ and $d_{i2}$ are linear combinations of symbols stored in the first and the second parity nodes. Clearly, the first term containing $\bm{m}_i$ in (\ref{Sysrep}) is the desired ``signal'' and the other term is ``interference'' which should be aligned. As the surviving two systematic nodes $j$ contain $\bm{m}_j$, $j\neq i$, the newcomer can further download data from these two nodes to cancel out the ``interference'' in (\ref{Sysrep}). However, to achieve minimum repair bandwidth, the newcomer can only download one symbol per node. To perform optimal repair through interference alignment, the following conditions must hold:
\begin{equation} \label{IAsys}
\begin{array}{l}
\text{rank} ( \left[
\begin{array}{c}
\bm{\varphi}_{i1}^T \\
\bm{\varphi}_{i2}^T \bm{A}_i^T \\
\end{array}
\right] ) = 2 \\
\text{rank} ( \left[
\begin{array}{c}
\bm{\varphi}_{i1}^T \\
\bm{\varphi}_{i2}^T \bm{A}_j^T \\
\end{array}
\right] ) = 1, ~~ j \neq i
\end{array}
\end{equation}
As the newcomer wants to recover $\bm{m_i}$ stored in failed nodes, the matrix before $\bm{m_i}$ in (\ref{Sysrep}) should have full rank so that $\bm{m_i}$ might be repaired by pre-multiplying the inverse of this matrix after interference cancellation. This also implies that $\bm{\varphi}_{i1}^T$ should have full row rank. The other conditions then actually require that the row spaces of $\bm{\varphi}_{i1}^T$ and $\bm{\varphi}_{i2}^T \bm{A}_j^T$ to fully overlap. Under these conditions, the newcomer can download only one linear combination of symbols at each systematic node, which can be either $\bm{\varphi}_{i1}^T \bm{m}_j$ or $\bm{\varphi}_{i2}^T \bm{A}_j^T \bm{m}_j$, to (align with and) cancel the interference.

Now in order to show that the systematic nodes can be optimally repaired, we only need to prove that there exist $\bm{\varphi}_{i1}^T$ and $\bm{\varphi}_{i2}^T$, $\forall i=1,2,3$, such that the conditions in (\ref{IAsys}) hold. This is established in the following three lemmas.

\begin{figure}
\centering
\includegraphics[width=2.5in]{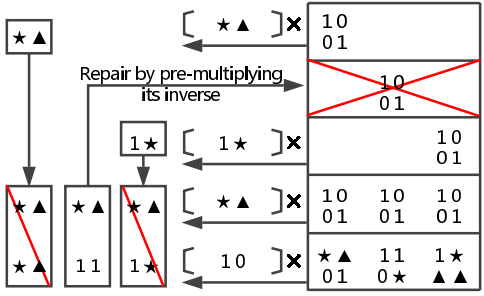}
\caption{Repair of the second systematic node}
\end{figure}

\begin{lemma} \label{3}
For $i = 3$, the conditions in (\ref{IAsys}) hold with $\bm{\varphi}_{32}^T =\bm{e}_2^T $ and $\bm{\varphi}_{31}^T = \bm{\varphi}_{32}^T \bm{A}_1^T$.
\end{lemma}
\begin{proof}
Clearly $\bm{e}_1^T  \bm{\varphi}_{32} = 0$ and $\bm{e}_2^T \bm{\varphi}_{32} = 1$. Then we have
\begin{equation}
\begin{array}{l}
\bm{A}_1 \bm{\varphi}_{32} = \left[
    \begin{array}{c}
    \theta \bm{e}_1^T \\
    \eta \bm{e}_1^T + \bm{e}_2^T \\
    \end{array}
    \right] \bm{\varphi}_{32} = \left[
    \begin{array}{c}
    0 \\
    1 \\
    \end{array}
    \right]
\\
\bm{A}_2 \bm{\varphi}_{32}
  = \frac{1}{\lambda} \left[
    \begin{array}{c}
    (\theta -1) \bm{e}_1^T \\
    \eta \bm{e}_1^T + \bm{e}_2^T \\
    \end{array}
    \right]  \bm{\varphi}_{32}
    = \frac{1}{\lambda} \left[
    \begin{array}{c}
    0 \\
    1 \\
    \end{array}
    \right]
\\
\bm{A}_3 \bm{\varphi}_{32}
  = \frac{1}{\mu} \left[
    \begin{array}{c}
    \theta \bm{e}_1^T - \bm{e}_2^T \\
    \eta \bm{e}_1^T + \bm{e}_2^T \\
    \end{array}
    \right] \bm{\varphi}_{32} =
    \frac{1}{\mu} \left[
    \begin{array}{c}
    - 1 \\
    1 \\
    \end{array}
    \right]
    \end{array}
\end{equation}
Since $\bm{\varphi}_{31}^T = \bm{\varphi}_{32}^T \bm{A}_1^T = \bm{e}_2^T \bm{A}_1^T$, and $\lambda, \mu \neq 1$ in (\ref{A}), we further have
\begin{equation}
    \begin{array}{l}
    \text{rank} (
    \left[
    \begin{array}{c}
    \bm{\varphi}_{31}^T \\
    \bm{\varphi}_{32}^T \bm{A}_1^T \\
    \end{array}
    \right] ) = \text{rank} (
    \left[
    \begin{array}{cc}
    0 & 1 \\
    0 & 1 \\
    \end{array}
    \right] ) = 1
    \\
    \text{rank} (
    \left[
    \begin{array}{c}
    \bm{\varphi}_{31}^T \\
    \bm{\varphi}_{32}^T \bm{A}_2^T \\
    \end{array}
    \right] ) = \text{rank}
    (
    \left[
    \begin{array}{cc}
    0 & 1 \\
    0 & 1/\lambda \\
    \end{array}
    \right] ) = 1 \\
    \text{rank}
    (
    \left[
    \begin{array}{c}
    \bm{\varphi}_{31}^T \\
    \bm{\varphi}_{32}^T \bm{A}_3^T \\
    \end{array}
    \right] ) = \text{rank}
     (
    \left[
    \begin{array}{cc}
    0 & 1 \\
    -1 / \mu & 1 / \mu \\
    \end{array}
    \right] ) = 2
\end{array}
\end{equation}
The proof completes.
\end{proof}

\begin{lemma} \label{2}
For $i = 2$, the conditions in (\ref{IAsys}) hold with $\bm{\varphi}_{22}^T =\bm{e}_1^T $ and $\bm{\varphi}_{21}^T = \bm{\varphi}_{22}^T \bm{A}_1^T$.
\end{lemma}
\begin{proof}
As $\bm{e}_1^T  \bm{\varphi}_{22} = 1$ and $\bm{e}_2^T \bm{\varphi}_{22} = 0$, we have
\begin{equation}
\begin{array}{l}
\bm{A}_1 \bm{\varphi}_{22} = \left[
    \begin{array}{c}
    \theta \bm{e}_1^T  \\
    \eta \bm{e}_1^T + \bm{e}_2^T \\
    \end{array}
    \right] \bm{\varphi}_{22} = \left[
    \begin{array}{c}
    \theta \\
    \eta \\
    \end{array}
    \right]
\\
\bm{A}_2 \bm{\varphi}_{22}
   = \frac{1}{\lambda} \left[
    \begin{array}{c}
    (\theta -1) \bm{e}_1^T \\
    \eta \bm{e}_1^T + \bm{e}_2^T \\
    \end{array}
    \right] \bm{\varphi}_{22}
    = \frac{1}{\lambda} \left[
    \begin{array}{c}
    \theta -1 \\
    \eta \\
    \end{array}
    \right]
\\
\bm{A}_3 \bm{\varphi}_{22}
   = \frac{1}{\mu} \left[
    \begin{array}{c}
    \theta \bm{e}_1^T - \bm{e}_2^T \\
    \eta \bm{e}_1^T + \bm{e}_2^T \\
    \end{array}
    \right] \bm{\varphi}_{22} =
    \frac{1}{\mu} \left[
    \begin{array}{c}
    \theta \\
    \eta \\
    \end{array}
    \right]
\end{array}
\end{equation}
Since $\bm{\varphi}_{21}^T = \bm{\varphi}_{22}^T \bm{A}_1^T = \bm{e}_1^T \bm{A}_1^T$, $\lambda, \mu \neq 0$ in (\ref{A}), and $\theta,\eta \neq 0$ in (\ref{limi-c})--(\ref{limi-d}), we can derive
\begin{equation}
    \begin{array}{l}
    \text{rank}
    (
    \left[
    \begin{array}{c}
    \bm{\varphi}_{21}^T \\
    \bm{\varphi}_{22}^T \bm{A}_1^T \\
    \end{array}
    \right]
    ) = \text{rank} (
    \left[
    \begin{array}{cc}
    \theta & \eta \\
    \theta & \eta \\
    \end{array}
    \right]) = 1,
\\
    \text{rank}
    (
    \left[
    \begin{array}{c}
    \bm{\varphi}_{21}^T \\
    \bm{\varphi}_{22}^T \bm{A}_3^T \\
    \end{array}
    \right] ) = \text{rank}(
    \left[
    \begin{array}{cc}
    \theta & \eta \\
    \theta / \mu & \eta / \mu \\
    \end{array}
    \right] )= 1,
\\
    \text{rank} (
    \left[
    \begin{array}{c}
    \bm{\varphi}_{21}^T \\
    \bm{\varphi}_{22}^T \bm{A}_2^T \\
    \end{array}
    \right]
    ) = \text{rank}
     (
    \left[
    \begin{array}{cc}
    \theta & \eta \\
    \frac{\theta - 1}{\lambda} & \eta / \lambda \\
    \end{array}
    \right] ) = 2.
\end{array}
\end{equation}
The lemma follows.
\end{proof}

\begin{figure}
\centering
\includegraphics[width=2.5in]{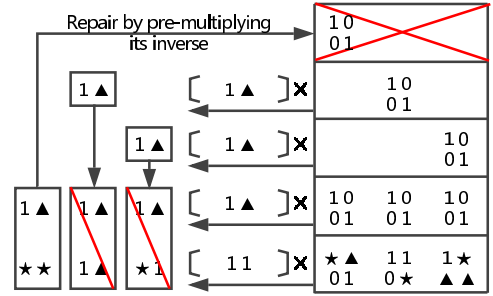}
\caption{Repair of the first systematic node}
\end{figure}

\begin{lemma} \label{1}
For $i = 1$, the conditions in (\ref{IAsys}) hold with $\bm{\varphi}_{12}^T =\bm{e}_1^T+ \bm{e}_2^T$ and $\bm{\varphi}_{11}^T = \bm{\varphi}_{12}^T \bm{A}_2^T$.
\end{lemma}
\begin{proof}
It is clear that $\bm{e}_1^T \bm{\varphi}_{12} = \bm{e}_2^T \bm{\varphi}_{12} = 1$. Then we have
\begin{equation}
\begin{array}{l}
\bm{A}_1 \bm{\varphi}_{12} = \left[
    \begin{array}{c}
    \theta \bm{e}_1^T  \\
    \eta \bm{e}_1^T + \bm{e}_2^T \\
    \end{array}
    \right] \bm{\varphi}_{12} = \left[
    \begin{array}{c}
    \theta \\
    \eta + 1 \\
    \end{array}
    \right] \\
\bm{A}_2 \bm{\varphi}_{12} = \frac{1}{\lambda} \left[
    \begin{array}{c}
    (\theta-1) \bm{e}_1^T \\
    \eta \bm{e}_1^T + \bm{e}_2^T \\
    \end{array}
    \right] \bm{\varphi}_{12} = \frac{1}{\lambda} \left[
    \begin{array}{c}
    \theta - 1 \\
    \eta + 1 \\
    \end{array}
    \right] \\
\bm{A}_3 \bm{\varphi}_{12} = \frac{1}{\mu} \left[
    \begin{array}{c}
    \theta \bm{e}_1^T - \bm{e}_2^T\\
    \eta \bm{e}_1^T + \bm{e}_2^T \\
    \end{array}
    \right] \bm{\varphi}_{12} = \frac{1}{\mu} \left[
    \begin{array}{c}
    \theta - 1 \\
    \eta + 1 \\
    \end{array}
    \right]
\end{array}
\end{equation}
Since $\lambda, \mu \neq 0$, $\theta \neq 0$ in (\ref{limi-c}) and $\eta \neq -1$ in (\ref{limi-d}), we have
\begin{equation}
\begin{array}{c}
    \text{rank}
    (
    \left[
    \begin{array}{c}
    \bm{\varphi}_{11}^T \\
    \bm{\varphi}_{12}^T \bm{A}_2^T \\
    \end{array}
    \right]
    )
    = \text{rank} (
    \left[
    \begin{array}{cc}
    \frac{\theta - 1}{\lambda} & \frac{\eta + 1}{\lambda} \\
    \frac{\theta - 1}{\lambda} & \frac{\eta + 1}{\lambda} \\
    \end{array}
    \right]) = 1, \\
    \text{rank}
    (
    \left[
    \begin{array}{c}
    \bm{\varphi}_{11}^T \\
    \bm{\varphi}_{12}^T \bm{A}_3^T \\
    \end{array}
    \right]
    )
    = \text{rank}
    (
    \left[
    \begin{array}{cc}
    \frac{\theta - 1}{\lambda} & \frac{\eta + 1}{\lambda} \\
    \frac{\theta - 1}{\mu} & \frac{\eta + 1}{\mu} \\
    \end{array}
    \right]) = 1, \\
    \text{rank} (
    \left[
    \begin{array}{c}
    \bm{\varphi}_{11}^T \\
    \bm{\varphi}_{12}^T \bm{A}_1^T \\
    \end{array}
    \right] ) = \text{rank} (
     \left[
    \begin{array}{cc}
    \frac{\theta - 1}{\lambda} & \frac{\eta + 1}{\lambda} \\
    \theta & \eta + 1 \\
    \end{array}
    \right] ) = 2.
\end{array}
\end{equation}
\end{proof}

Based on Lemmas III.1--III.3, we readily arrive at
\begin{theorem}
Any systematic node can be repaired with minimum repair bandwidth.
\end{theorem}

Use the proposed MDS codes over the finite field $\mathbb{F}_4$ in Tab. II as an example. Fig. 2 illustrates the optimal repair of the third systematic node, where we use $\blacktriangle$ to denote $w$ and $\bigstar$ to denote $w+1$. The generator matrix is then
\[
\bm{P} = \left[
        \begin{array}{cccccc}
        1 & 0 &   &   &   &  \\
        0 & 1 &   &   &   &   \\
          &   & 1 & 0 &   &   \\
          &   & 0 & 1 &   &   \\
          &   &   &   & 1 & 0 \\
          &   &   &   & 0 & 1 \\
        1 & 0 & 1 & 0 & 1 & 0 \\
        0 & 1 & 0 & 1 & 0 & 1 \\
        \bigstar & \blacktriangle & 1 & 1 & 1 & \bigstar \\
        0 & 1 & 0 & \bigstar & \blacktriangle & \blacktriangle \\
        \end{array}
\right]
\]
With $\bm{m} := [\bm{m}_1^{T},\; \bm{m}_2^{T}, \;\bm{m}_3^{T}]^{T}$, when the third systematic node fails,the newcomer needs to recover the symbols stored in it (i.e.,  symbols in $\bm{m}_3$). To this end, the newcomer downloads one symbol from each surviving node. Specifically, it uses coefficient vector $\bm{\varphi}_{31}^T = \bm{e}_3^T \bm{A}_1^T=[0 \;\; 1]$ to obtain a linear combination of two symbols at the first parity node:
\[
    d_{31}=\bm{\varphi}_{31}^T \bm{u_4} = [0 \;\; 1] \bm{m}_1 + [0 \;\; 1] \bm{m}_2 + [0 \;\; 1] \bm{m}_3.
\]
Similarly, the newcomer uses coefficient vector $\bm{\varphi}_{32}^T =[0 \;\; 1]$ to obtain a linear combination of two symbols at the second parity node:
\[
    d_{32}=\bm{\varphi}_{32}^T \bm{u_5} = [0 \;\; 1] \bm{m}_1 + [0 \;\; \bigstar] \bm{m}_2 + [\blacktriangle \;\; \blacktriangle] \bm{m}_3.
\]
Then the newcomer downloads one symbol $f_{31}:=[0 \;\; 1]\bm{m}_1$ and $f_{32}:=[0 \;\; 1]\bm{m}_2$ from the first and two systematic nodes, respectively. With the downloaded symbols, the newcomer can cancel out the interference and obtain
\[
\left[
    \begin{array}{c}
    d_{31} - f_{31} - f_{32} \\
    d_{32} - f_{31} - \bigstar \cdot f_{32} \\
    \end{array}
\right] =
\left[
    \begin{array}{cc}
    0 & 1 \\
    \blacktriangle & \blacktriangle \\
    \end{array}
    \right] \bm{m}_3.
\]
Since the matrix before $\bm{m}_3$ has full-rank, the desired symbols (i.e. symbols in $\bm{m}_3$) can be retrieved by pre-multiplying the inverse of this matrix.

Fig. 3 and Fig. 4 illustrate the optimal repair of the second and the first systematic node, respectively. Similar procedures can be performed to repair these two nodes using minimum bandwidth.

\subsection{Repair of the First Parity Node}

Different from systematic node case, interference alignment cannot be used directly for the repair of a parity node. Hence, we perform a change of basis as follows. Let
\begin{equation}
\begin{array}{lcl}
\bm{u} & = & \bm{P} \cdot \bm{m} = \left[
       \begin{matrix}
        \bm{I} & & \\
         & \bm{I} &\\
         & & \bm{I}\\
        \bm{I} & \bm{I} & \bm{I}\\
        \bm{A}_1^T & \bm{A}_2^T & \bm{A}_3^T\\
        \end{matrix}
        \right] \bm{m} \\
        & = & \left[
        \begin{matrix}
        \bm{I} & & \\
         & \bm{I} &\\
        -\bm{I} & -\bm{I} & \bm{I}\\
         & & \bm{I}\\
        \bm{A}_1^T - \bm{A}_3^T & \bm{A}_2^T - \bm{A}_3^T & \bm{A}_3^T\\
        \end{matrix}
        \right] \bm{m}' \\
\end{array}
\end{equation}
where
\[
\bm{m}' = \left[
        \begin{matrix}
        \bm{I} & & \\
         & \bm{I}&\\
        \bm{I} & \bm{I} & \bm{I}\\
        \end{matrix}
        \right] \bm{m}.
\]
Now the first parity node looks like a systematic node and the third systematic node looks like a parity node. To optimally repair the first parity node, the newcomer needs to download one symbol from each surviving node. Define $\bm{\varphi}_{41}^T$ and $\bm{\varphi}_{42}^T$, which are $1 \times 2$ row vectors used to access symbols from the third systematic node and the second parity node, respectively. The newcomer then obtains
\begin{equation}
\begin{array}{l}
\bm{\varphi}_{41}^{T}
    \left[
    \begin{array}{ccc}
    -\bm{I} & -\bm{I} & \bm{I}\\
    \end{array}
    \right] \bm{m}'
\\
\bm{\varphi}_{42}^{T}
    \left[
    \begin{array}{ccc}
    \bm{A}_1^T - \bm{A}_3^T & \bm{A}_2^T - \bm{A}_3^T & \bm{A}_3^T \\
    \end{array}
    \right] \bm{m}'.
\end{array}
\end{equation}
As with (\ref{IAsys}), the following conditions must be satisfied to perform interference alignment using the minimum repair bandwidth.
\begin{equation} \label{First}
\begin{array}{l}
    \text{rank}
    (
    \left[
    \begin{array}{c}
    \bm{\varphi}_{41}^T \\
    \bm{\varphi}_{42}^T (\bm{A}_1^T - \bm{A}_3^T) \\
    \end{array}
    \right] ) = 1
\\
    \text{rank}
    (
    \left[
    \begin{array}{c}
    \bm{\varphi}_{41}^T \\
    \bm{\varphi}_{42}^T (\bm{A}_2^T - \bm{A}_3^T) \\
    \end{array}
    \right] ) = 1
\\
    \text{rank}
    (
    \left[
    \begin{array}{c}
    \bm{\varphi}_{41}^T \\
    \bm{\varphi}_{42}^T \bm{A}_3^T \\
    \end{array}
    \right] ) = 2
\end{array}
\end{equation}

\begin{figure}
\centering
\includegraphics[width=2.5in]{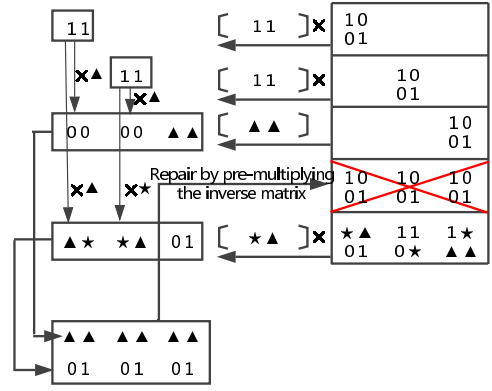}
\caption{Repair of the first parity node}
\end{figure}

We prove the following lemma.
\begin{lemma} \label{1st}
The conditions in (\ref{First}) hold with $\bm{\varphi}_{42}^T =(\mu -1)\bm{e}_2^T +(\lambda -1)\bm{e}_1^T$ and $\bm{\varphi}_{41}^T = \bm{\varphi}_{42}^T (\bm{A}_1^T - \bm{A}_3^T)$.
\end{lemma}
\begin{proof}
It is easy to show:
\begin{equation} \label{First2}
\begin{array}{lcl}
\bm{A}_1 \bm{\varphi}_{42} & = & \left[
    \begin{array}{c}
    \theta \bm{e}_1^T  \\
    \eta \bm{e}_1^T + \bm{e}_2^T \\
    \end{array}
    \right] \bm{\varphi}_{42} \\
    & = & \left[
    \begin{array}{c}
    \theta (\lambda - 1) \\
    \eta (\lambda - 1) + \mu - 1 \\
    \end{array}
    \right], \\
\bm{A}_2 \bm{\varphi}_{42} & = & \frac{1}{\lambda} \left[
    \begin{array}{c}
    (\theta-1) \bm{e}_1^T \\
    \eta \bm{e}_1^T + \bm{e}_2^T \\
    \end{array}
    \right] \bm{\varphi}_{42} \\
    & = & \frac{1}{\lambda} \left[
    \begin{array}{c}
    (\theta - 1) (\lambda - 1) \\
    \eta (\lambda - 1) + \mu - 1 \\
    \end{array}
    \right], \\
\bm{A}_3 \bm{\varphi}_{42} & = & \frac{1}{\mu} \left[
    \begin{array}{c}
    \theta \bm{e}_1^T - \bm{e}_2^T\\
    \eta \bm{e}_1^T + \bm{e}_2^T \\
    \end{array}
    \right] \bm{\varphi}_{42} \\
    & = & \frac{1}{\mu} \left[
    \begin{array}{c}
    \theta (\lambda - 1) - (\mu - 1) \\
    \eta (\lambda - 1) + \mu - 1 \\
    \end{array}
    \right].
\end{array}
\end{equation}
Recall that $\bm{A}_1^T = \lambda \bm{A}_2^T + [\bm{e}_1 \;\; \bm{0}] = \mu \bm{A}_3^T + [\bm{e}_2 \;\; \bm{0}]$. Hence, we can further derive
\begin{equation}
\begin{array}{lcl}
(\bm{A}_1 - \bm{A}_3) \bm{\varphi}_{42} & = & \frac{\lambda(\mu - 1)}{\mu - \lambda} (\bm{A}_2 - \bm{A}_3) \bm{\varphi}_{42} \\
 & & + \left[
    \begin{array}{c}
    \frac{\mu - 1}{\mu - \lambda}\bm{e}_1^T + \frac{1 - \lambda}{\mu - \lambda} \bm{e}_2^T \\
    \bm{0}^T \\
    \end{array}
    \right] \bm{\varphi}_{42} \\
    & = & \frac{\lambda(\mu - 1)}{\mu - \lambda} (\bm{A}_2 - \bm{A}_3) \bm{\varphi}_{42}.
    \end{array}
\end{equation}
In addition, we can use (\ref{First2}) to get
\begin{equation}
(\bm{A}_1 - \bm{A}_3) \bm{\varphi}_{42} = \frac{\mu - 1}{\mu} \left[
    \begin{array}{c}
    \theta (\lambda - 1) + 1 \\
    \eta (\lambda - 1) + \mu - 1 \\
    \end{array}
    \right].
\end{equation}
As $\mu \neq 0, 1$ and $(\mu - 1) \neq \eta (1- \lambda)$ by (\ref{limi-i}), i.e., $ \eta (\lambda - 1) + (\mu - 1) \neq 0 $, we have $\text{rank} (\bm{\varphi}_{42}^T (\bm{A}_1^T - \bm{A}_3^T)) = 1 $. Hence,
\begin{equation}
\begin{array}{l}
\text{rank}
    (
    \left[
    \begin{array}{c}
    \bm{\varphi}_{41}^T \\
    \bm{\varphi}_{42}^T (\bm{A}_1^T - \bm{A}_3^T) \\
    \end{array}
    \right] ) \\
    \quad = \text{rank}
    (
    \left[
    \begin{array}{c}
    \bm{\varphi}_{42}^T (\bm{A}_1^T - \bm{A}_3^T) \\
    \bm{\varphi}_{42}^T (\bm{A}_1^T - \bm{A}_3^T) \\
    \end{array}
    \right] ) = 1,
\\
    \text{rank}
    (
    \left[
    \begin{array}{c}
    \bm{\varphi}_{41}^T \\
    \bm{\varphi}_{42}^T (\bm{A}_2^T - \bm{A}_3^T) \\
    \end{array}
    \right] ) \\
    \quad = \text{rank} (
    \left[
    \begin{array}{c}
    \bm{\varphi}_{42}^T (\bm{A}_1^T - \bm{A}_3^T) \\
    \frac{\mu - \lambda}{\lambda (\mu - 1)}\bm{\varphi}_{42}^T (\bm{A}_1^T - \bm{A}_3^T) \\
    \end{array}
    \right] ) = 1 .
\end{array}
\end{equation}
Since $\bm{A}_1 = \mu \bm{A}_3 + [\bm{e}_2 \;\; \bm{0}]^T$ (\ref{A123}), we have
\begin{equation}
\begin{array}{ll}
    & \text{rank} (
    \left[
    \begin{matrix}
    \bm{\varphi}_{41}^T \\
    \bm{\varphi}_{42}^T \bm{A}_3^T \\
    \end{matrix}
    \right] )
     =  \text{rank} (
    \left[
    \begin{matrix}
    (\bm{A}_1 - \bm{A}_3)\bm{\varphi}_{42} & \bm{A}_3 \bm{\varphi}_{42} \\
    \end{matrix}
    \right] ) \\
    &~~ =  \text{rank} (
    \left[
    \begin{matrix}
    (\bm{A}_1 - \bm{A}_3)\bm{\varphi}_{42} & [\bm{e}_2 \;\; \bm{0}]^T \bm{\varphi}_{42} \\
    \end{matrix}
    \right] ) \\
    & ~~= \text{rank} (
    \displaystyle \frac{\mu - 1}{\mu} \left[
    \begin{matrix}
    \theta (\lambda - 1)+1 & \mu \\
    (\lambda -1) \eta + (\mu - 1) & 0 \\
    \end{matrix}
    \right] ) = 2 .
\end{array}
\end{equation}
The last step holds since $\eta (\lambda - 1) + (\mu - 1) \neq 0$ and $\mu \neq 0, 1$.
\end{proof}

\begin{theorem}
The first parity node can be optimally repaired using minimum bandwidth.
\end{theorem}
\begin{proof}
By Lemma \ref{1st}, we have found $\bm{\varphi}_{41}^T$ and $\bm{\varphi}_{42}^T$ satisfying (\ref{First}). Hence, through downloading one symbol from each surviving node and performing interference alignment, the first parity node can be optimally repaired.
\end{proof}

Again, construct the proposed MDS codes over the finite field $\mathbb{F}_4$ in Tab. II. Fig. 5 illustrates the optimal repair of the first parity node with the same notations as in Fig. 2. Now the newcomer needs to recover the symbols stored in the first parity node, by downloading one symbol from each surviving node. It uses coefficient vector $\bm{\varphi}_{41}^T =[\blacktriangle \;\; \blacktriangle]$ to obtain a linear combination of two symbols at the third systematic node:
\[
    d_{41}=\bm{\varphi}_{41}^T \bm{u_4} = [\blacktriangle \;\; \blacktriangle] \bm{m}_3;
\]
and uses coefficient vector $\bm{\varphi}_{42}^T = [\bigstar \;\; \blacktriangle]$ to obtain a linear combination of two symbols at the second parity node:
\[
    d_{42}=\bm{\varphi}_{42}^T \bm{u_5} = [\blacktriangle \;\; \bigstar] \bm{m}_1 + [\bigstar \;\; \blacktriangle] \bm{m}_2 + [0 \;\; 1] \bm{m}_3.
\]
Then the newcomer downloads one symbol $f_{41}:=[1 \;\; 1]\bm{m}_1$ and $f_{42}:=[1 \;\; 1]\bm{m}_2$ from the first and two systematic nodes, respectively. With the downloaded symbols, the newcomer can align the interference and obtain
\[
\left[
    \begin{array}{c}
    d_{41} + \blacktriangle \cdot f_{41} + \blacktriangle \cdot f_{42} \\
    d_{42} + \blacktriangle \cdot f_{41} + \bigstar \cdot f_{42} \\
    \end{array}
\right] =
\left[
    \begin{array}{cc}
    \blacktriangle & \blacktriangle \\
    0 & 1 \\
    \end{array}
    \right] \sum_{i = 1,2,3} \bm{m}_i.
\]
Since the matrix before $\sum_{i = 1,2,3} \bm{m}_i$ has full-rank, the desired symbols (i.e. symbols in the first parity node, which are $\bm{u}_4=\bm{m}_1 + \bm{m}_2 + \bm{m}_3$) can be retrieved by pre-multiplying the inverse of this matrix.

\begin{figure}
\centering
\includegraphics[width=2.5in]{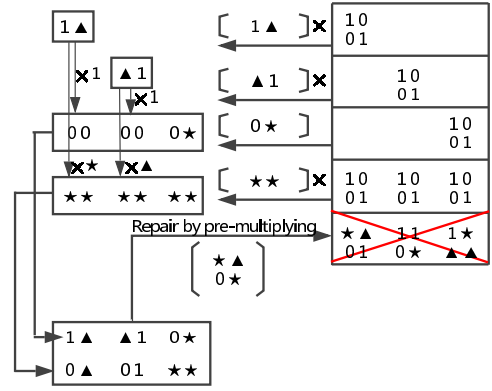}
\caption{Repair of the second parity node}
\end{figure}

\subsection{Repair of the Second Parity Node}

As with repairing the first parity node, a change of basis is needed. Since we have proven the invertibility of $\bm{A}_1^T$, $\bm{A}_2^T$ and $\bm{A}_3^T$ in Lemma \ref{invert}, we have
\begin{equation}
\begin{array}{l}
\bm{u} = \bm{P} \cdot \bm{m} = \left[
       \begin{matrix}
        \bm{I} & & \\
         & \bm{I} &\\
         & & \bm{I}\\
        \bm{I} & \bm{I} & \bm{I}\\
        \bm{A}_1^T & \bm{A}_2^T & \bm{A}_3^T\\
        \end{matrix}
        \right] \bm{m} \\
        = \left[
        \begin{matrix}
        \bm{I} & & \\
         & \bm{I} &\\
        -(\bm{A}_3^T)^{-1} \bm{A}_1^T & -(\bm{A}_3^T)^{-1} \bm{A}_2^T & (\bm{A}_3^T)^{-1} \\
        \bm{I} -(\bm{A}_3^T)^{-1} \bm{A}_1^T & \bm{I} -(\bm{A}_3^T)^{-1} \bm{A}_2^T & (\bm{A}_3^T)^{-1} \\
         & & \bm{I}\\
        \end{matrix}
        \right] \bm{m}'', \\
\end{array}
\end{equation}
where
\[
\bm{m}'' = \left[
        \begin{array}{ccc}
        \bm{I} & & \\
         & \bm{I} &\\
        \bm{A}_1^T & \bm{A}_2^T & \bm{A}_3^T\\
        \end{array}
        \right] \bm{m} .
\]
Now the second parity node looks like a systematic node. To optimally repair the second parity node, we need to access one symbol from each surviving node and interference alignment would be performed to repair the failed node.

Let $\bm{\varphi}_{51}^T = \bm{\varphi}_{51}^{'T} \bm{A}_3^T$ and $\bm{\varphi}_{52}^T = \bm{\varphi}_{52}^{'T} \bm{A}_3^T$, where $\bm{\varphi}_{51}^T$, $\bm{\varphi}_{51}^{'T}$, $\bm{\varphi}_{52}^T$ and $\bm{\varphi}_{52}^{'T}$ are all $1 \times 2$ row vectors. Using $\bm{\varphi}_{51}^T$ and $\bm{\varphi}_{52}^T$, the newcomer obtains two symbols from the third systematic node and the first parity node as follows:
\begin{equation}
\begin{array}{l}
    \bm{\varphi}_{51}^{T}
    \left[
    \begin{matrix}
    -(\bm{A}_3^T)^{-1} \bm{A}_1^T & -(\bm{A}_3^T)^{-1} \bm{A}_2^T & (\bm{A}_3^T)^{-1} \\
    \end{matrix}
    \right] \bm{m}'' \\ \quad = \bm{\varphi}_{51}^{'T} \left[
    \begin{matrix}
    -\bm{A}_1^T & -\bm{A}_2^T & \bm{I} \\
    \end{matrix}
    \right] \bm{m}'',
\\
    \bm{\varphi}_{52}^{T}
    \left[
    \begin{matrix}
    \bm{I} -(\bm{A}_3^T)^{-1} \bm{A}_1^T & \bm{I} -(\bm{A}_3^T)^{-1} \bm{A}_2^T & (\bm{A}_3^T)^{-1} \\
    \end{matrix}
    \right] \bm{m}'' \\ \quad= \bm{\varphi}_{52}^{'T} \left[
    \begin{matrix}
    \bm{A}_3^T - \bm{A}_1^T & \bm{A}_3^T - \bm{A}_2^T & \bm{I} \\
    \end{matrix}
    \right] \bm{m}''.
\end{array}
\end{equation}
The following equations should be satisfied in order to perform optimal repair of the second parity node [cf. (\ref{IAsys})]:
\begin{equation} \label{Second}
\begin{array}{l}
    \text{rank}
    (
    \left[
    \begin{array}{c}
    \bm{\varphi}_{51}^{'T} \bm{A}_1^T \\
    \bm{\varphi}_{52}^{'T} (\bm{A}_3^T - \bm{A}_1^T) \\
    \end{array}
    \right] ) = 1
\\
    \text{rank}
    (
    \left[
    \begin{array}{c}
    \bm{\varphi}_{51}^{'T} \bm{A}_2^T \\
    \bm{\varphi}_{52}^{'T} (\bm{A}_3^T - \bm{A}_2^T) \\
    \end{array}
    \right] ) = 1
\\
    \text{rank}
    (
    \left[
    \begin{array}{c}
    \bm{\varphi}_{51}^{'T} \\
    \bm{\varphi}_{52}^{'T} \\
    \end{array}
    \right] ) = 2
\end{array}
\end{equation}
To satisfy the last equality in (\ref{Second}), $\bm{\varphi}_{51}^{'T}$ and $\bm{\varphi}_{52}^{'T}$ must be non-zero vectors. By Lemmas \ref{invert} and \ref{MDS}, we have proven the invertibility of $\bm{A}_i^T$ and $\bm{A}_i^T-\bm{A}_j^T$, $(i \neq j)$. Hence, $\bm{\varphi}_{51}^{'T} \bm{A}_1^T$, $\bm{\varphi}_{51}^{'T} \bm{A}_2^T$, $\bm{\varphi}_{52}^{'T} (\bm{A}_3^T - \bm{A}_1^T)$ and $\bm{\varphi}_{52}^{'T} (\bm{A}_3^T - \bm{A}_2^T)$ are all non-zero vectors. In order to satisfy the first two equalities in (\ref{Second}), we need
\begin{equation} \label{lin}
\begin{array}{c}
\bm{A}_1 \bm{\varphi}_{51}^{'} = c (\bm{A}_3 - \bm{A}_1) \bm{\varphi}_{52}^{'}\\
\bm{A}_2 \bm{\varphi}_{51}^{'} = \tilde{c} (\bm{A}_3 - \bm{A}_2) \bm{\varphi}_{52}^{'},
\end{array}
\end{equation}
where $c, \tilde{c} \neq 0$. That is,
\begin{equation} \label{p1}
\bm{A}_1^{-1} (\bm{A}_3 - \bm{A}_1) \bm{\varphi}_{52}^{'} = p \bm{A}_2^{-1} (\bm{A}_3 - \bm{A}_2) \bm{\varphi}_{52}^{'},
\end{equation}
where $p = \tilde{c} / c$. Rewrite the equation as
\begin{equation} \label{p}
((\bm{A}_1^{-1} \bm{A}_3 - \bm{I}) - p(\bm{A}_2^{-1} \bm{A}_3 - \bm{I}))\bm{\varphi}_{52}^{'} = \bm{0},
\end{equation}
where $\bm{0}$ denotes $2 \times 2$ zero matrix. To obtain a non-zero solution $\bm{\varphi}_{52}^{'}$ for (\ref{p}), we need
\begin{equation} \label{assu2}
\text{rank} ((\bm{A}_1^{-1} \bm{A}_3 - \bm{I}) - p(\bm{A}_2^{-1} \bm{A}_3 - \bm{I})) = 1.
\end{equation}

Following the aforementioned lines, we can then find $\bm{\varphi}_{51}^{'}$ and $\bm{\varphi}_{52}^{'}$ satisfying (\ref{Second}) as follows:
\begin{enumerate}
    \item Find $p \neq 0$ satisfying (\ref{assu2}).
    \item Obtain a non-zero solution $\bm{\varphi}_{52}^{'}$ for (\ref{p}) such that (\ref{p1}) holds.
    \item Let $\bm{\varphi}_{51}^{'} = \bm{A}_1^{-1} (\bm{A}_3 - \bm{A}_1) \bm{\varphi}_{52}^{'}$; then (\ref{lin}) as well as the first two equalities in (\ref{Second}) hold .
    \item Check if the last equality in (\ref{Second}) holds.
\end{enumerate}

The following three lemmas show that the aforementioned procedure can successfully obtain the desired $\bm{\varphi}_{51}^{'}$ and $\bm{\varphi}_{52}^{'}$.

\begin{lemma} \label{pex}
Condition (\ref{assu2}) holds with $p= \frac{\mu - 1}{\mu - \lambda}$.
\end{lemma}
\begin{proof}
First, we have
\begin{equation}
\begin{array}{lcl}
    & & \text{rank}((\bm{A}_1^{-1} \bm{A}_3 - \bm{I}) - p(\bm{A}_2^{-1} \bm{A}_3 - \bm{I})) \\
    & = & \text{rank}(\bm{A}_1 \cdot ((\bm{A}_1^{-1} \bm{A}_3 - \bm{I}) - p(\bm{A}_2^{-1} \bm{A}_3 - \bm{I}))) \\
    & = & \text{rank}((p-1)\bm{A}_1 + \bm{A}_3 - p \bm{A}_1 \bm{A}_2^{-1} \bm{A}_3)
\end{array}
\end{equation}
By (\ref{A}),
\begin{equation}
\bm{A}_2 = \frac{1}{\lambda} \left[
    \begin{array}{c}
    (\theta -1) \bm{e}_1^T \\
    \eta \bm{e}_1^T + \bm{e}_2^T\\
    \end{array}
    \right].
\end{equation}
Denote the inverse of $\bm{A}_2$ by $ \bm{A}_2^{-1} := [\bm{b}_1 \quad \bm{b}_2]$, where $\bm{b}_1$ and $\bm{b}_2$ are $2 \times 1$ vectors. By $\bm{A}_2 \bm{A}_2^{-1} = \bm{I}$, we readily have
\begin{equation}
\begin{array}{rcl}
\bm{e}_1^T  \bm{b}_1 & = & \frac{\lambda}{\theta - 1}, \\
\bm{e}_1^T \bm{b}_2 & = & 0 ,\\
(\eta \bm{e}_1^T + \bm{e}_2^T) \bm{b}_1 & = & 0, \\
(\eta \bm{e}_1^T + \bm{e}_2^T) \bm{b}_2 & = & \lambda. \\
\end{array}
\end{equation}
Therefore,
\begin{equation}
\begin{array}{rcl}
\bm{A}_1 \bm{A}_2^{-1} \bm{A}_3 & = & \left[
\begin{array}{c}
\theta \bm{e}_1^T \\
\eta \bm{e}_1^T + \bm{e}_2^T \\
\end{array}
\right] [\bm{b}_1 \quad \bm{b}_2] \bm{A}_3 \\
& = & \left[
\begin{array}{cc}
\frac{\theta \lambda}{\theta - 1} & 0 \\
0 & \lambda \\
\end{array}
\right] \frac{1}{\mu} \left[
\begin{array}{c}
\theta \bm{e}_1^T - \bm{e}_2^T \\
\eta \bm{e}_1^T + \bm{e}_2^T \\
\end{array}
\right] \\
& = & \frac{\lambda}{\mu} \left[
    \begin{array}{c}
    \frac{\theta}{\theta -1}(\theta \bm{e}_1^T - \bm{e}_2^T ) \\
    \eta \bm{e}_1^T + \bm{e}_2^T \\
    \end{array}
    \right]
\end{array}
\end{equation}
With $p= \frac{\mu - 1}{\mu - \lambda}$, we have
\begin{equation}
\begin{array} {lcl}
    & & (p-1)\bm{A}_1 + \bm{A}_3 - p \bm{A}_1 \bm{A}_2^{-1} \bm{A}_3 \\
   & = & \frac{\lambda - 1}{\mu - \lambda}  \left[
        \begin{array}{c}
        \theta \bm{e}_1^T \\
        \eta \bm{e}_1^T + \bm{e}_2^T \\
        \end{array}
        \right] + \frac{1}{\mu} \left[
        \begin{array}{c}
        \theta \bm{e}_1^T - \bm{e}_2^T \\
        \eta \bm{e}_1^T + \bm{e}_2^T \\
        \end{array}
        \right] \\
        & &  - \frac{ \lambda (\mu - 1)}{\mu (\mu - \lambda)} \left[
        \begin{array}{c}
        \frac{\theta}{\theta -1}(\theta \bm{e}_1^T - \bm{e}_2^T ) \\
        \eta \bm{e}_1^T + \bm{e}_2^T \\
        \end{array}
        \right] \\
   & = & \left[
        \begin{array}{c}
        \frac{-\lambda \theta (\mu - 1)}{\mu(\mu - \lambda)(\theta - 1)} \bm{e}_1^T + (\frac{-1}{\mu} +  \frac{1}{\mu} \frac{\mu - 1}{\mu - \lambda} \frac{\theta \lambda}{\theta - 1}) \bm{e}_2^T\\
        \bm{0}^T \\
        \end{array}
        \right] \\
\end{array}
\end{equation}
Hence, $\text{rank} ((\bm{A}_1^{-1} \bm{A}_3 - \bm{I}) - p(\bm{A}_2^{-1} \bm{A}_3 - \bm{I})) = \text{rank}((p-1)\bm{A}_1 + \bm{A}_3 - p \bm{A}_1 \bm{A}_2^{-1} \bm{A}_3) = 1$.
\end{proof}

\begin{lemma} \label{IF}
For two vectors $\bm{v}_1=c_1 \bm{e}_1 + d_1 \bm{e}_2$ and $\bm{s}_2=c_2 \bm{e}_1 + d_2 \bm{e}_2$, where $c_1$, $c_2$, $d_1$, $d_2$ are elements over the finite field $\mathbb{F}_q$ and $c_1 d_2 - c_2 d_1 \neq 0$, then we have $\bm{v}_1^T \bm{\omega} = 0 $ and $\bm{v}_2^T \bm{\omega} \neq 0$ with $\bm{\omega} = c_1 \bm{e}_2 - d_1 \bm{e}_1$.
\end{lemma}
\begin{proof}
This can be shown by direct calculations:
\begin{equation}
\begin{array}{lcl}
\bm{v}_1^T \bm{\omega} & = & (c_1 \bm{e}_1^T + d_1 \bm{e}_2^T)(c_1 \bm{e}_2 - d_1 \bm{e}_1) \\
& = & - c_1 d_1 + d_1 c_1 = 0 ,\\
\bm{v}_2^T \bm{\omega} & = & (c_2 \bm{e}_1^T + d_2 \bm{e}_2^T)(c_1 \bm{e}_2 - d_1 \bm{e}_1) \\
& = & c_1 d_2 - c_2 d_1  \neq 0.
\end{array}
\end{equation}
\end{proof}

\begin{lemma} \label{sec}
There exist $\bm{\varphi}_{51}^{'}$ and $\bm{\varphi}_{52}^{'}$ satisfying (\ref{Second}).
\end{lemma}
\begin{proof}
Let $c_1 = \frac{-\lambda \theta (\mu - 1)}{\mu(\mu - \lambda)(\theta - 1)}$, $d_1 = \frac{-1}{\mu} + \frac{\mu - 1}{\mu - \lambda} \frac{\theta \lambda}{\theta - 1} \frac{1}{\mu}$, $c_2 = \eta$ and $d_2 = 1$. Then
\begin{equation}
\begin{array}{lcl}
& & c_1 d_2 - c_2 d_1 \\
& = & \frac{-\lambda \theta (\mu - 1)}{\mu(\mu - \lambda)(\theta - 1)} - \eta (\frac{-1}{\mu} + \frac{\mu - 1}{\mu - \lambda} \frac{\theta \lambda}{\theta - 1} \frac{1}{\mu}) \\
& = & \frac{\lambda \theta (\mu - 1)}{\mu(\mu - \lambda)(\theta - 1)} (\eta \frac{(\theta -1)(\mu - \lambda)}{\theta \lambda(\mu - 1)} - (\eta + 1))
\end{array}
\end{equation}
By (\ref{limi-a})--(\ref{limi-c}) and (\ref{limi-j}), we can derive $c_1 d_2 - c_2 d_1 \neq 0$. Let $\bm{v}_1 = c_1 \bm{e}_1 + d_1 \bm{e}_2$, $\bm{v}_2 = c_2 \bm{e}_1 + d_2 \bm{e}_2$. In the proof of Lemma \ref{pex}, we show that when $p= \frac{\mu - 1}{\mu - \lambda}$,
\begin{equation}
(p-1)\bm{A}_1 + \bm{A}_3 - p \bm{A}_1 \bm{A}_2^{-1} \bm{A}_3 = \left[
\begin{array}{c}
    \bm{v}_1^T \\
    \bm{0}^T \\
\end{array}
\right],
\end{equation}
It follows from Lemma \ref{IF} that if $c_1 d_2 - c_2 d_1 \neq 0$, then $\bm{\omega} = c_1 \bm{e}_2 - d_1 \bm{e}_1$ satisfies  $\bm{v}_1^T \bm{\omega} = 0$ and $\bm{v}_2^T \bm{\omega} \neq 0$.

Let $\bm{\varphi}_{52}^{'} = \bm{\omega}$, then $\bm{v}_1^T \bm{\varphi}_{52}^{'} = \bm{v}_1^T  \bm{\omega} = 0$ and $\bm{v}_2^T \bm{\varphi}_{52}^{'} = \bm{v}_2^T \bm{\omega} \neq 0$. Hence, $\bm{\varphi}_{52}^{'}$ is in fact a non-zero solution for (\ref{p}) since
\begin{equation}
\begin{array}{lcl}
& & ((\bm{A}_1^{-1} \bm{A}_3 - \bm{I}) - p(\bm{A}_2^{-1} \bm{A}_3 - \bm{I}))\bm{\varphi}_{52}^{'} \\
 & = & \bm{A}_1^{-1}\left[
\begin{array}{c}
    \bm{v}_1^T \\
    \bm{0}^T \\
\end{array}
\right] \bm{\varphi}_{52}^{'} = \bm{0}
\end{array}
\end{equation}

Let $\bm{\varphi}_{51}^{'} = \bm{A}_1^{-1}(\bm{A}_3 - \bm{A}_1) \bm{\varphi}_{52}^{'}$. Then the first two equalities in (\ref{Second}) are met. Let $c_2^{'} = 0$, $d_2^{'} = 1$, and $\bm{v}_2^{'} = c_2^{'} \bm{e}_1 + d_2^{'} \bm{e}_2 = \bm{e}_2$. By (\ref{limi-a})--(\ref{limi-c}), we derive
\begin{equation}
c_1 d_2^{'} - c_2^{'} d_1 = c_1 = \frac{-\lambda \theta (\mu - 1)}{\mu(\mu - \lambda)(\theta - 1)} \neq 0
\end{equation}
Again, Lemma \ref{IF} implies that $\bm{\omega} = c_1 \bm{e}_2 - d_1 \bm{e}_1$ satisfies $\bm{v}_2^{'T} \bm{\omega} = \bm{e}_2^T \bm{\varphi}_{52}^{'} \neq 0$, as $c_1 d_2^{'} - c_2^{'} d_1 \neq 0$. Since we also have $\bm{A}_1^T=\mu \bm{A}_3^T + [\bm{e}_2 \;\; \bm{0}]$ in (\ref{A123}) and $\bm{v}_2^T \bm{\varphi}_{52}^{'} = \bm{v}_2^T \bm{\omega} \neq 0$, it readily follows.
\begin{equation}
\begin{array}{ll}
\text{rank}(
    \left[
    \begin{array}{c}
    \bm{\varphi}_{51}^{'T} \\
    \bm{\varphi}_{52}^{'T} \\
    \end{array}
    \right] ) & = \text{rank}(\left[
    \begin{array}{c}
    \bm{\varphi}_{52}^{'T} (\bm{A}_3 - \bm{A}_1)^T (\bm{A}_1^T)^{-1}\\
    \bm{\varphi}_{52}^{'T} \\
    \end{array}
    \right]) \\
    & = \text{rank}(\left[
    \begin{array}{c}
    \bm{\varphi}_{52}^{'T} (\bm{A}_3^T - \bm{A}_1^T)\\
    \bm{\varphi}_{52}^{'T} \bm{A}_1^T \\
    \end{array}
    \right]) \\
     & = \text{rank}(\left[
    \begin{array}{c}
    \bm{\varphi}_{52}^{'T} \bm{A}_3^T\\
    \bm{\varphi}_{52}^{'T} \bm{A}_1^T \\
    \end{array}
    \right]) \\
    & = \text{rank}(\left[
    \begin{array}{c}
    \bm{\varphi}_{52}^{'T} [\bm{e}_2 \;\; \bm{0}]\\
    \bm{\varphi}_{52}^{'T} \bm{A}_1^T \\
    \end{array}
    \right]) \\
     & = \text{rank}(
    \left[
    \begin{array}{cc}
    \bm{\varphi}_{52}^{'T} \cdot \bm{v}_2^{'}  & 0 \\
    \theta \bm{\varphi}_{52}^{'T} \cdot \bm{x} &  \bm{\varphi}_{52}^{'T} \cdot \bm{v}_2 \\
    \end{array}
    \right] ) \\
    & = 2;
    \end{array}\nonumber
\end{equation}
i.e., the last equality in (\ref{Second}) is satisfied
\end{proof}

\begin{theorem}
The second parity node can be optimally repaired.
\end{theorem}
\begin{proof}
By Lemma \ref{sec}, we have found proper $\bm{\varphi}_{51}^{'}$ and $\bm{\varphi}_{52}^{'}$ satisfying (\ref{Second}). Hence, through downloading one symbol from each surviving node and performing interference alignment, the second parity node can be optimally repaired.
\end{proof}

Fig. 6 illustrates the optimal repair of the second parity node, with the proposed MDS codes constructed over the finite field $\mathbb{F}_4$ in Tab. II. The operations are similar to those with repair of the first parity node.

\subsection{Discussions}

Theorems III.4, III.6, and III.10 together establish that the proposed MDS codes have optimal repair capability; i.e., they can ensure exact repair of every single node failure using minimum bandwidth. With the proposed MDS codes, the node repairs require low-complexity computations. As shown in Fig.2 to Fig.6, the newcomer downloads one linear combination of the $\alpha=2$ symbols from each of the $n-1$ surviving nodes; the relevant computation complexity is ${\cal O}((n-1)\alpha)$. The interference cancellation requires a computational complexity ${\cal O}((k-1)\alpha/2)$. After canceling out all the interference, the desired data can be repaired by multiplying an $\alpha \times \alpha$ matrix, with a computation complexity ${\cal O}(\alpha^{2})$ for repairing systematic nodes, and ${\cal O}(M \alpha)$ for repairing parity nodes. Since the proposed codes have a minimum stripe size $\alpha=2$ and are constructed over $\mathbb{F}_4$, the computational complexity with node repairs is clearly very low.

Recall that the proposed MDS codes can be only used to store a data file of $M=6$ symbols. However, following the simple duplication technique in \cite{Product}, we can employ the proposed codes to store a data file with more symbols. For example, we first construct two DSS (denoted by DSS0 and DSS1) using the proposed MDS codes, and then combine them into a new DSS (denoted by DSS2). The first node of DSS0 contains two symbols denoted by $m_{11}$ and $m_{12}$, and the first node of DSS1 contains two symbols denoted by $m_{11}'$ and $m_{12}'$, then the first node of DSS2 contains four symbols: $m_{11}$, $m_{12}$, $m_{11}'$ and $m_{12}'$. Constructions of other nodes in DSS2 are similar. By this way, we can double the storage capacity of DSS and all the properties remains. Assume two nodes fail in DSS2, we can still retrieve the original file using the surviving nodes. Because the symbols in these nodes are able to retrieve the original file of DSS0 and DSS1 respectively, they can also retrieve the original file of DSS2, which is the combination of files in DSS0 and DSS1. Similarly, optimal node repairs can be also performed after duplication.

%

\section{Conclusions}

We developed a novel construction of $(5,3)$ MDS codes. Constructed over a small finite field $\mathbb{F}_4$, our codes have the minimum stripe size $\alpha = 2$, and ensure optimal repair of every single node failure using minimum bandwidth. It is certainly interesting to extend our approach to construct general $(k+2,k)$ MDS codes with small stripe size and optimal repair capability. This will be pursued in future work.


\end{document}